\documentclass[11pt]{article}

\usepackage{amsmath,amssymb,amsfonts}
\usepackage{graphicx}
\usepackage{latexsym}
\usepackage{theorem}
\usepackage{subfigure}
\usepackage{float}
\usepackage{algorithmic}
\usepackage[ruled]{algorithm}
\usepackage{citesort}

\theoremstyle{plain}
\theorembodyfont{\rm}
\newtheorem{theorem}{Theorem}
\newtheorem{lemma}{Lemma}

\newtheorem{proof}{Proof}
\newtheorem{definition}{Definition}
\newtheorem{example}{Example}

\newcommand{\ot}{\ensuremath{\leftarrow}}

\newcommand{\Omit}[1]{}

\newcounter{Codeline}
\newcommand{\Newcodeline}{\setcounter{Codeline}{1}}
\newcommand{\Cl}{{\theCodeline}: \addtocounter{Codeline}{1}}
\newcommand{\crm}{\\}

\renewcommand{\algorithmicrequire}{\textbf{when}}

\newcommand{\trans}[1]{\ensuremath{\stackrel{#1}{\longrightarrow}}}

\newcommand{\qed}{\hfill $\Box$}

\newcommand{\Scomment}[1]{\ensuremath{/\ast} #1 \ensuremath{\ast/}}

\newcommand{\Real}{\mathbb{R}}

\begin{document}

\begin{titlepage}
\vspace*{10mm}
\begin{center}
\begin{minipage}{170mm}
\begin{center}
{\LARGE The BG-simulation for Byzantine Mobile Robots}\footnotemark[1] \\
\end{center} 
\end{minipage} \\
\vspace*{10mm}
Taisuke Izumi\footnotemark[2]\footnotemark[4] \hspace{3mm} 
Zohir Bouzid\footnotemark[3] \hspace{3mm}
S\'{e}bastien Tixeuil\footnotemark[3] \hspace{3mm}
Koichi Wada\footnotemark[2]
\end{center}

\footnotetext[1]{This work was supported in part by ANR projects,
KAKENHI no.21500013 and no.22700010.}
\footnotetext[2]{Nagoya Institute of Technology, Gokiso-cho, 
Showa-ku, Nagoya, Aichi, 466-8555, Japan. E-mail: \{t-izumi,wada\}@nitech.ac.jp.}
\footnotetext[3]{Universit\'{e} Pierre et Marie Curie - Paris 6, LIP6 CNRS 7606, 
France. E-mail: zohir.bouzid@gmail.com, Sebastien.Tixeuil@lip6.fr.}
\footnotetext[4]{Corresponding Author. E-mail: t-izumi@nitech.ac.jp. 
Tel:+81-52-735-5567, Fax:+81-52-735-5408}

\begin{abstract}
This paper investigates the task solvability of mobile robot systems subject 
to Byzantine faults. We first consider the \emph{gathering} problem, which 
requires all robots to meet in finite time at a non-predefined location. 
It is known that the solvability of Byzantine gathering strongly depends 
on a number of system attributes, such as synchrony, the number of Byzantine 
robots, scheduling strategy, obliviousness, orientation of local coordinate 
systems and so on. However, the complete characterization of the attributes 
making Byzantine gathering solvable still remains open.

In this paper,  we show strong impossibility results of Byzantine gathering. 
Namely, we prove that Byzantine gathering is impossible even if we assume 
one Byzantine fault, an atomic execution system, the $n$-bounded centralized 
scheduler, non-oblivious robots, instantaneous movements and a common 
orientation of local coordinate systems (where $n$ denote the number 
of correct robots). Those hypotheses are much 
weaker than used in previous work, inducing a much stronger impossibility 
result. 

At the core of our impossibility result is a reduction from the distributed 
consensus problem in asynchronous shared-memory systems. In more details, 
we newly construct a generic reduction scheme based on the distributed 
{\em BG-simulation}. Interestingly, because of its versatility, we can 
easily extend our impossibility result for general pattern formation 
problems. 



\end{abstract}
\end{titlepage}

\section{Introduction}

\paragraph{Motivation}

Robot networks have recently become a challenging research area for distributed 
computing researchers. At the core of scientific studies lies the 
characterization of the minimum robots capabilities that are necessary to 
achieve non-trivial tasks, such as the formation of geometric patterns, 
scattering, gathering, \emph{etc}. The considered robots are often very weak: 
They are anonymous (\emph{i.e.} that do not have any means to perform distinct 
tasks based on a distinguishable identifier), oblivious (\emph{i.e.} they cannot 
remember past observations, computations, or movements), disoriented (\emph{i.e.} 
they share neither a common coordinate system nor a common length unit), and 
most importantly dumb (\emph{i.e.} they don't have any explicit mean of 
communication). The last property means that robots cannot communicate 
explicitly by sending messages to one another. Instead, their communication is 
indirect (or spatial): a robot 'writes' a value to the network by moving toward 
a certain position, and a robot 'reads' the state of the network by observing 
the positions of other robots in terms of its {\em local coordinate system}.

The problem we consider in this paper is the \emph{gathering} of 
fault-prone robots~\cite{SY99}. Given a set of oblivious 
robots with arbitrary initial locations and no agreement on a 
global coordinate system, the gathering problem requires that 
all correct robots reach and stabilize the same, but unknown 
beforehand, location. A number of solvability issues about
the gathering problem are studied in previous works because 
of its fundamental importance in both theory and practice.
One can easily find an analogy of the gathering problem 
to the {\em consensus} problem, and thus may think that 
its solvability issue are straightforwardly deduced from the 
known results about the consensus solvability (e.g., FLP impossibility). 
However, many differences lies between those two problems and
the solvability of the gathering problem is still non-trivial.
We can enumerate at least three factors that strongly affect
the solvability of the gathering problem: \emph{(i)} the absence of a common 
coordinate system, \emph{(ii)} the fact that there is no explicit termination,
and \emph{(iii)} the lack of a validity requirement. In fault-free
 environments, the non-triviality of the existence of a solution mainly 
results from \emph{(i)} that hardens symmetry breaking. Actually, gathering 
is known to be impossible to solve with 
$n=2$ robots in atomic-execution (ATOM) models\footnote{While ATOM models 
are often called semi-synchronous models, we do not use that word
because this model actually has no bound for the processing/moving 
speed of each robot. We adapt the notion of bounded schedulers for
characterizing the bound for processing speed of robots, and thus 
apply the word ``asynchronous ATOM models'' to the conventional
semi-synchronous models.}. 
One direction of the study of gathering is to explore the weaker 
assumptions breaking this hardness. For example, 
endowing robots with a small amount of 
memory~\cite{SY99,BDPT10}, or weak agreement of local coordinate
systems~\cite{SDY06,IKIW07,ITIKW08}. 
On the other hand, in fault-prone environments, the remaining two factors
arise as the primary difference to the consensus in classical computation
models. An important witness encouraging the difference is 
that the gathering problem can be solved in a certain kind of crash-prone 
asynchronous robot networks~\cite{agmonP06,defago3274fta}, while 
the consensus cannot be solved under the asynchrony and one 
crash fault~\cite{fischer1985impossibility}.

\paragraph{Our Contribution}

In this paper, we investigate the solvability of the gathering problem
in robot networks subject to {\em Byzantine} faults. While
crash-faulty robots just stop the execution of the deployed algorithm, 
a Byzantine-faulty robot may execute arbitrary code (including malicious 
code) and try to defeat the proper operation of correct robots. 
As we mentioned, the solvability of Byzantine gathering is quite
non-trivial. Actually, the Byzantine-tolerant gathering problem still 
has the large gap between possibility and impossibility. As known 
results, Byzantine gathering is feasible only under very strong 
assumptions (fully-synchronous ATOM models or small number of robots)~\cite{agmonP06}, and also the impossibility results are 
proved only for severe
models (asynchrony, oblivious and uniform robots, and/or without agreement 
of coordinate systems)~\cite{agmonP06,defago3274fta}. 
Filling this gap has remained 
an open question until now. In this paper, we respond negatively: Namely, 
we prove that Byzantine gathering is impossible even if we assume an 
ATOM models, $n$-bounded centralized scheduler, non-oblivious and 
non-uniform robots, and a common orientation of local coordinate 
systems, for only one Byzantine robot (where $n$ denotes the number of 
correct robots). Those assumptions are much stronger than that shown 
in previous work, inducing a much stronger impossibility result.

At the core of our impossibility result is a reduction to 
$1$-Byzantine-resilient gathering in mobile robot systems from 
the distributed $1$-crash-resilient consensus problem in 
asynchronous shared-memory systems. In 
more details, based on the distributed BG-simulation by Borowsky 
and Gafni~\cite{borowsky1993generalized,BGLR01}, we newly 
construct a $1$-crash-resilient consensus algorithm using any $1$-Byzantine-resilient gathering algorithm on the system with several 
constraints. Thus, we can deduce impossibility results of 
Byzantine gathering for the model stated above. 
More interestingly, because of its versatility, we can easily extend our 
impossibility result for general pattern formation problems: We show 
that the impossibility also holds for a broad class of pattern 
formation problems including line and circle formation. To the best of our knowledge, this 
paper is the first study explicitly bridging algorithmic mobile robotics 
and conventional distributed computing theory for proving impossibility 
results.

It is remarkable that we assume a certain kind of synchrony assumption 
for robot systems. The assumption of $n$-bounded scheduler restricts 
the relative speed of each robot (formally, $n$-bounded scheduler 
only allows the activation schedules where each robot is activated 
at most $n$ times between any two consecutive activations of 
some robot). An interesting insight we can find from our result is that 
it is possible to trade the synchrony and Byzantine behavior of robot 
networks to the asynchrony and crash behavior of shared memory systems,
which implies that the gap between synchronous robot networks
and classical distributed computation models is as large as that
between synchrony and asynchrony in classical models. 

\paragraph{Related works}

Since the pioneering work of Suzuki and Yamashita~\cite{SY99}, the 
formation of a specific patterns, including the gathering and the 
convergence problems, by mobile robots has been addressed first 
in \emph{fault-free} systems for a broad class of settings. 
Prencipe~\cite{Pre05} studied the problem of gathering in
both atomic and non-atomic movement models, and showed that the problem 
is unsolvable without additional assumptions such as being able to detect 
the multiplicity of a location (\emph{i.e.}, knowing if there is more 
than one robot in a given location). Following their work, the gathering 
and the convergence problems were considered on several restricted settings 
such as with limited visibility~\cite{FPSW05,ando1999dmp}, and 
with inaccurate sensors and 
movements~\cite{SDY06,YIKIW09,IKIW07,ITIKW08,CohenP08}.

The case of \emph{fault-prone} robot networks was recently tackled by several 
academic studies. The faults that have been investigated fall in two 
categories: \emph{crash} faults and \emph{Byzantine} faults. The 
\emph{deterministic} fault-tolerant gathering is first 
addressed in~\cite{agmonP06} where the authors propose a gathering 
algorithm that tolerates one crash in ATOM models with arbitrary schedulers and
another algorithm working under the fully synchronous scheduling, which tolerates up to $f$ Byzantine faults for $n > 2f$ robot systems, where
$n$ is the number of correct robots. In \cite{defago3274fta}, the authors 
study the feasibility of \emph{probabilistic} gathering in crash-prone 
and Byzantine-prone environments. It also improves the impossibility of
Byzantine gathering, but the impossibility still relies on 
the weakness of models, including obliviousness and no agreement of 
coordinate systems. 

The convergence problem,
which is a variation of the gathering problem, was first addressed 
by Cohen and Peleg~\cite{CohenP08}, where algorithms based on convergence to 
the center of gravity of the system are presented. Those algorithms work in
non-atomic models with asynchronous schedulers.
The study of convergence in Byzantine-prone environments are addressed
by Bouzid \emph{et al.} A series of their papers~\cite{Bouzid2010,BGT09e} investigates
the relationship between the maximum number of faulty robots and the 
synchrony and the atomicity of robots. 

As impossibility results are hard to get, it is often interesting to start from 
a small set of such impossibility results and derive others through 
\emph{reduction}. The distributed {\em BG-simulation} lies as one of the 
powerful reduction schemes  in distributed computing. There are 
many applications of it with a variety of modified reduction
strategies~\cite{borowsky1993generalized,BGLR01,GafniGP05,Gafni09}.

\paragraph{Roadmap}
The organization of the paper is as follows: In section \ref{secPreliminary}, we 
explain the system model, including both robot and shared-memory models, and the 
problem definitions. Section \ref{secReduction} introduces our reduction scheme. 
To clarify the concept of our idea, this section shows a weaker version of the 
reduction, which is be extended and generalized  at Section 
\ref{secGeneralPattern}. 

\section{Preliminaries}
\label{secPreliminary}
\subsection{Asynchronous Shared Memory System}

We consider a single-writer multi-reader (SWMR) asynchronous shared memory 
system of $m_{S}$ processes $\{p_0, \cdots, p_{m_{S}-1}\}$. The 
shared memory consists of a number of memory cells, one of which can 
be atomically read and written by each process. That is, we assume 
{\em linearizable} shared memory. We also employ {\em atomic snapshot} 
for access to the shared memory. It provides atomic read of all shared 
memory cells. Since atomic snapshot operation can be implemented 
by only using read/write operations, it gives no additional 
computational power to the model. Any impossibility result on
asynchronous shared memory systems also holds even if we assume atomic 
snapshot operation. Since all of three operations to access the shared 
memory are performed atomically and instantaneously, we see it as an 
{\em event} in execution. To explain the order of events easily, we 
often use the notion of discrete global time. Each event is assigned 
the global timestamp of its occurrence. Since we assume linearizability, 
all events are consistently serialized. Thus, without loss of generality,
we assume any two events necessarily have different timestamps. 
Note that the global time is introduced only for ease of explanation, 
and no process is aware of the time.

Processes are subject to crash faults. When a process crashes,
it stops all of the following operations and becomes silent.
In this paper, we assume that only one process can be crashed. 

\subsection{Consensus Problem}

In a consensus algorithm, each correct process initially proposes a value, 
and eventually chooses a decision value from the values proposed 
by processes so that all processes decide the same value. 
The standard specification of the consensus problem assumes that the tree following properties are satisfied: \emph{(i)} {\bf Termination} Every correct process eventually decides, \emph{(ii)} {\bf Agreement} No two correct processes decide different values, and \emph{(iii)} {\bf Validity} If a process decides a value $v$, then, $v$ is a value proposed by a process.

Throughout this paper, we only consider the binary consensus,
where only value zero or one is the possible proposal.
It is well-known that the consensus problem is not solvable in
asynchronous shared-memory systems with one crash fault~\cite{LA87}.

\begin{theorem}[Impossibility of $1$-resilient Consensus] \label{thmImpConsensus}
There is no binary consensus algorithm on the asynchronous 
SWMR shared-memory model even if $m_{S} = 2$ and only one process can be crashed.
\end{theorem}

\subsection{Byzantine Mobile Robot System}

The robot system consists of $n + 1$ autonomous mobile robots 
$\mathcal{R} = \{ r_0, \cdots, r_{n+1}\}$ for $n > 2$. 
Each robot is non-oblivious (it can memorize a history of execution) and 
can be non-uniform (all robots can execute different codes)\footnote{
Note that our non-uniformity does not means that each robot has 
a different identity and any other robot can read this identity from
the observation. That is, each robot has a different identity but it 
can not visibly identify the labels of other robots. These difference
are inherently important. Actually, it yields a different computational
power to the robot system\cite{das10}.}
\footnote{Note that the model of this paper is stronger than the standard 
one, which usually assumes that each robot is anonymous, oblivious and uniform. 
However our aim is proving the impossibility, and thus to assuming stronger 
assumptions gives a stronger impossibility result.}

It does not have any device for direct communication, but is capable of 
observing its environment (\emph{i.e.}, the positions of other robots in its 
local coordinate system). One robot is modeled as a point located on a 
two-dimensional space. To specify the location of each robot consistently, we 
use a global Cartesian coordinate system. Notice that this global coordinate 
system is introduced only to ease the explanations, and that each robot is not 
aware of it. Each robot executes the deployed algorithm in {\em computational 
cycles} (or briefly {\em cycles}). At the beginning of a cycle, the robot 
observes the current environment (\emph{i.e.}, the positions of other robots) 
and determines the destination point based on the deployed algorithm. Then, 
the robot moves toward the computed destination\footnote{We also assume 
that the robot can reach the computed destination in the move phase for 
proving the impossibility}, which concludes the cycle. 

The local coordinate system of a robot is the Cartesian coordinate system whose 
origin is the current position of the robot. Moreover, $x$ and $y$ axes of each 
robot are parallel, \emph{i.e.} robots share a common direction. We assume 
strong multiplicity detection for the observation of points with two or more 
robots: Each robot can detect the exact number of robots that are located at a 
particular point. 

We assume the ATOM execution model, where an execution is divided into consecutive {\em rounds}. The \emph{scheduler} determines the set of performing robots for each round. At any round $r = 0, 1, 2, \cdots$, the scheduler determines whether each robot is {\em active} or {\em inactive}. Active robots perform one cycle in an atomic manner, and inactive ones wait during the round. The scheduler is \emph{fair} in the sense that every robot is activated infinitely often. In this paper, we assume that the scheduler is also {\em $k$-bounded}, which guarantees that if a robot is activated at round $r_1$ and $r_2$ ($r_1 < r_2$), any robot is activated at most $k$ times during $[r_1, r_2]$. We also consider another constraint for the scheduler, called {\em centralized} scheduler, which allows only one robot to be activated at each round.

In our model, robot may exhibit Byzantine faults. A Byzantine robot 
is allowed not to follow the deployed algorithm, and thus behaves
arbitrarily. However, if we consider the $k$-bounded scheduler, 
the constraint is also incurred to faulty robots: Even a Byzantine 
robot may change its position at most $k$ times during 
two consecutive activations of any correct robot. We call robots 
that are not Byzantine {\em correct}. Throughout this paper, we assume
that the system has one Byzantine robot.

In what follows, we give a formalisms of the robot model 
we now consider: Let $\mathcal{S}$ be the set of all possible internal 
states of the algorithm $\{\mathcal{A}_0, \mathcal{A}_1, \cdots 
\mathcal{A}_n\}$, where $\mathcal{A}_i$ is the algorithm deployed to $r_i$. 
We can define the {\em local state} of a robot
as a pair of its current internal state and its location in 
terms of global coordinate systems. A {\em system 
configuration} (or {\em configuration} for short), is an 
$(n+1)$-tuple of local states where each corresponds to 
the local configuration of a robot. We also define a location
configuration $L(C)$ of $C$ as an $(n+1)$-tuple of the global 
coordinates each of which corresponds to locations of each robot
at $C$. We sometimes treat $L(C)$ as a multiset. The location 
of robot $r_i$ at $C$ is denoted by $\mathbf{r}_i(C)$.
Algorithm $\mathcal{A}_i$ is defined as a mapping
$\mathcal{A}_i : \{\Real^2\}^{n+1} \times \mathcal{S} \to \Real^2 \times
\mathcal{S}$. That is, each robot computes the destination and its 
poststate from the observation result (i.e., the multiset of locations
in terms of its local coordinate system) and the current internal state. 

An {\em execution} of an algorithm is a sequence of configurations
$C_0, C_1, C_2, \cdots  $ where $C_{i+1}$ can be obtained from $C_i$
by making a number of robots move following the deployed algorithm
and by changing the location of Byzantine robot arbitrarily.


\subsection{Gathering Problem}
The {\em gathering problem} must ensure that all correct 
robots eventually meet at a point that is not predefined, starting 
from any configuration. Formally, we say that an algorithm 
$\mathcal{A}$ solves the gathering problem if any execution of 
$\mathcal{A}$ eventually reaches a configuration where 
all correct robots are on a single point and never leave there. 

Given an execution $\mathcal{E} = C_0, C_1, \cdots, C_j, \cdots$,
we say a configuration $C_j$ is {\em legitimate} if all correct robots 
keeps a common location at $C_{j'}$ for any $j' > j$.
For any configuration $C_j$, we define $\mathbf{m}(C_j)$ as the point 
at which the most robots are located in $C_j$\footnote{If two or more points 
has the same and maximum number of robots, an arbitrary one from them 
is deterministically chosen as the value of $\mathbf{m}(C_j)$.}, 
and $M(C_j)$ as its number.

\section{Impossibility of Byzantine Gathering}
\label{secReduction}

\subsection{Discussion and Outline}

Our impossibility proof is derived from the reduction of 
$1$-crash-resilient binary consensus on asynchronous shared memory
systems from $1$-Byzantine-resilient gathering on mobile robot systems. 
More precisely, we show that the $1$-crash-resilient binary consensus
algorithm can be constructed using any $1$-Byzantine-resilient 
gathering algorithm. The primary part of this reduction is a 
simulator algorithm of Byzantine mobile robot systems, which is
``synchronous'' in a certain sense, on the top of 
``asynchronous'' shared-memory systems. 

Let us first start the explanation of our idea from an analogy 
bridging those two models. It is easy to find some correspondence 
between atomic snapshot models and mobile robots: We see one 
memory cell of the shared memory as the local 
state (i.e., the location) of one robot. Then, taking a snapshot 
implies an observation, and writing some value to the cell implies move. 
This analogy makes us look at a framework of the simulation as follows:
\begin{enumerate}
\item Consider the shared memory system of $n + 1$ processes, each of 
which corresponds to one simulated robot. At the beginning of 
the simulation, each process ``encodes'' its proposal to the 
initial location of the robot (e.g., put the robot on $(1, 0)$ 
if the proposal is one and on $(0, 0)$ if zero).
\item Each process $p_i$ repeats the following task: The process
$p_i$ first takes a snapshot. Using the resultant value of 
the snapshot as the observed configuration, it activates robot 
$r_i$ and calculates the destination of gathering algorithm. 
Then, it actually makes $r_i$ move to the computed destination by 
writing its coordinate to the corresponding shared-memory cell.
\item If the observed configuration is legitimate, each process 
``decodes'' the decision value from the coordinate of the gathering 
point. 
\end{enumerate}

Unfortunately, the above framework does not work correctly. We can point
out at least two flaws: 1) The above simulation framework is not 
wait-free: To simulate $(n+1)$-robot systems correctly, all $n+1$ 
robots must appear on any configuration. However, 
if a faulty process is initially crashed, the initial location 
is never set to the corresponding robot. 2) The observation and 
movement is not atomic: To simulate semi-synchronous mobile robots, any concurrent cycle 
must be performed in synchronized manner, which is not guaranteed in
the above framework. For example, the following behavior is 
possible: A process takes a snapshot at $t$, and writing the 
destination at $t'$ ($t < t'$). Then, during the period $[t, t']$,
another process may simulate two or more activations. That behavior 
never matches the ATOM execution. 

Our reduction circumvents the above difficulties by employing the
concept of the BG-simulation by Borowski and Gafni \cite{borowsky1993generalized}.
The BG-simulation is originally invented to extend the wait-free 
impossibility into $1$-resilient impossibility. Its principle is 
to simulate the system of $n$ processes with single fault by 
only two processes with single fault. To make such a simulation 
successful, we cannot statically assign the role of 
simulated processes to simulator processes (because all of the 
processes assigned to a process $p_i$ are simultaneously crashed 
if $p_i$ is crashed). Instead, in the original BG-simulation, 
each process simulates {\em all} processes in round-robin 
manner. Then, one process can be simulated by two processes, which
may brings some inconsistency problem between two simulations. 
The original BG-simulation resolves those situations by a 
mutual-exclusion-like mechanism. This approach is wait-free 
as the simulation of $1$-resilient systems: Assume that one 
of two simulating processes can be crashed when it is in 
simulation of a process $p_i$. Then, $p_i$'s simulation 
can be blocked forever, but such a block occurs at most 
once since the number of simulator processes is two. Thus, 
we can regard $p_i$ as faulty in the simulation. This feature
will be helpful to resolve the flaw 1) we mentioned.

Yet another problem of simulating ATOM, however, 
still remains even if we simply use the BG-simulation. 
This is because the original BG-simulation does not intend 
to simulate synchronous systems. To clarify this problem, 
let us consider  the following cases: Two simulator processes 
concurrently simulate the consecutive two behaviors $x$ and 
$y$ of two different processes. First, both processes simulate 
the behavior $x$ but their simulations are inconsistent. Then, 
to decide the behavior $x$, both processes must complete the 
simulation of $x$. However, since one of two simulator processes 
may be faulty, the process completing the simulation 
of $x$ cannot wait for the other process. It must proceed
to the simulation of $y$ in spite of uncertainty of $x$,
which can result in the inconsistency between $x$ and $y$
after both of them are fixed. The original BG-simulation avoids
this uncertainty by reading the past state at the simulation of
$y$. That is, any simulation following $x$ is performed as if
the behavior $x$ does not occur, which continues until 
the behavior $x$ is completely fixed. Then, importantly, 
we cannot know when $x$ is fixed because the simulator 
processes are completely asynchronous, and thus the 
simulation also becomes asynchronous. Consequently, the 
straightforward use of BG-simulation fails to achieve the 
simulation of our reduction. The trick of our reduction 
algorithm is to make Byzantine behavior absorb this 
uncertainty.

\subsection{Reduction}

\subsubsection{Object {\sf slot}}

Our reduction algorithm uses a {\sf slot} shared object, which partly
abstracts the idea of the original BG-simulation. Informally, a Slot 
object is the write-once register shared by two processes, 
which is guaranteed to decide one submitted value as the committed value
only if no process crashes during submission, or two submissions by
different processes are not contended.  
It provides two operations {\sf submit}$_i(v)$ and {\sf read}$_i()$. 
The operation {\sf submit}$_i(v)$ denotes that $p_i$ 
writes a value $v$ to the slot. Since {\sf slot} is write-once, 
it can be activated at most once by each process. The read operation 
by $p_i$ returns the triple $(v_0, v_1, s)$, which respectively mean the 
values submitted by $p_0$ and $p_1$, and the status of the object 
(if $v_i$ is not submitted yet, $v_i = \bot$). The status indicates whether
the stored value is committed or not, and which is the committed value 
if committed. If the value is committed, the status
entry in the returned triple holds the process ID submitting the committed
value. Otherwise, it holds value $\bot$, which means the slot is not
committed yet. Formally, we can define the specification of {\sf slot} 
object as follows:
\begin{definition}
Let $O$ be a {\sf slot} object. The time when operation $O.\sf{submit}_i(v_i)$ 
begins and ends is denoted by $b_i$ and $e_i$ \footnote{If the
operation $O.\sf{submit}_i(v_i)$ does not begin (or does not ends by
$p_i$'s crash), we define $b_i = \infty$ (or $e_i = \infty$).}.  
Then, the following properties are guaranteed:
\begin{description}
\item[Validity] For any triple $(w_0, w_1, s)$ returned by a read operation,
$w_i \in \{v_i, \bot\}$ holds for any $i \in \{0, 1\}$.
\item[Contended Value Detection] If a read operation returns 
$(w_0, w_1, \bot)$, $w_0 \neq \bot$ and $w_1 \neq \bot$ hold. 
\item[Persistency] If a read operation returns a non-$\bot$ status $s$ at $t$,
any read operation invoked after $t$ returns status $s$. 
\item[Commitment] Any read operation invoked at $t' > \max\{e_0, e_1\}$
returns a non-$\bot$ status. 
\item[No Contention Commitment] If $e_i < b_{1 -i}$, 
any read operation invoked after $e_i$ returns the status $s = i$.
\item[Common Value Commitment] If $v_0 = v_1$ holds, 
any read operation invoked at $t' > \min\{e_0, e_1\}$ returns 
a non-$\bot$ status.
\end{description}
\end{definition}

In this paper, we do not present the implementation of \textsf{slot} object
because it is implicitly addressed in the original BG-simulation paper 
\cite{borowsky1993generalized,BGLR01}. The readers who are interested in the implementation can refer 
that paper or a standard textbook of distributed computing.

\subsubsection{Details of Simulation}

Algorithm \ref{algo:CC} shows the pseudocode 
description of our simulation. As explained, this algorithm is 
designed for two-processor asynchronous shared memory systems. 
In the following argument, let $\{p_0, p_1\}$ be the set of
simulator processes running this algorithm. As the simulation target,
we consider the robot system of $f=1$. Hence the total number of 
robots is $n+1$ . In the simulation, $r_n$ is regarded as the Byzantine
robot. The shared memory has an array $E$ of \textsf{slot}
objects. Each element of $E$ stores the result of an activation 
of a robot (represented as a pair of its internal state and location), 
and the whole of array $E$ corresponds to the round-robin scheduling of 
all {\em correct} robots. Thus, each slot 
$E[j]$ stores the behavior of robot $r_{j \bmod n}$. Note that 
$E$ does not explicitly contain the behavior of faulty robot $r_n$. 

The simulation algorithm consists of two blocks: The first for-loop 
constructs the initial configuration of the simulated execution,
where each process submits the initial configuration of each robot with
location $(0, 0)$ or $(1, 0)$ to $E[0..n-1]$ according to the 
simulator's proposal value. Simulating one-step movement of a robot 
corresponds to one pass of the following loop block (referred as 
``main loop'' in the following argument). The variable $u$ counts 
the number of simulated steps. That is, $u$-th loop simulates 
$(u - n)$-th time step of the simulated execution. Recall that 
the first $n$ slots are used for the construction of the initial 
configuration.

In the loop, the simulation exploits the 
subroutine called \textsf{getview}. It constructs 
the configuration as the observation result of robot $r_{u \bmod n}$ 
by referring last $n$ slots: The subroutine first takes a snapshot 
of $E$ (referred as $E'$ in the algorithm), and copies the committed 
values of $E'[u - n], E'[u - n + 1], \cdots E'[u - 1]$ to local variable
$C[0], C[1], \cdots C[n - 1]$. If some slot $E[u - n + g]$ ($0 \leq g \leq 
n - 1$) is uncommitted, one cannot determine the value to be committed, but
only obtain two submitted values $v_0$ and $v_1$. Then, we store
$v_i$ into $C[g]$ ($i$ is the ID of the simulator process) and
$v_{1 - i}$ into $C[n]$. The implication of this scheme is 
to ``assume'' $v_i$ is the committed value and regard $v_{1 - i}$ 
as a Byzantine behavior. If there is no uncommitted slot, one does 
not have to use the Byzantine behavior for conflict resolution of 
uncommitted slot, and thus an arbitrary location can be given to 
Byzantine robot $r_n$. In our 
simulation, a ``helping'' location, which is $\mathbf{m}(C)$, is given. 
The \textsf{getview} subroutine also returns a flag $q$, which returns
TRUE if all slots of $E'[u - n.. u-1]$ are committed. This information is 
used to determine whether the simulator process can decide a value or not.
After the construction of the observation result $C$, if $q =$ TRUE and
the constructed configuration is legitimate. $p_i$ decides a value 
decoded from $\mathbf{m}(C)$. Note that this decode function cannot be
defined as $\mathsf{decode}((0, 0)) = 0$ and $\mathsf{decode}(\mathbf{v}) 
= 1$ for all other $\mathbf{v}$. Even if all correct robots are initially 
placed on a common point $\mathbf{v}$, the point of gathering is not 
necessarily $\mathbf{v}$ because we consider non-oblivious robots.
The way of defining function $\mathsf{decode}$ is argued in the 
following correctness proof.

\Newcodeline
\begin{algorithm}[h]
\caption{{\sf ConsensusToGathering}: Reduction from Consensus
to Byzantine Gathering}
\label{algo:CC}
{\small 
\begin{tabbing}
111 \= 11 \= 11 \= 11 \= 11 \= 11 \= 11 \= \kill
\Cl \> $E[0..\infty]$ of {\sf slot} (shared objects)\crm
\crm
\Cl \> {\bf procedure} {\sf getview}$_i(j)$: \crm
\Cl \> \> $q \ot$ TRUE;  $E' \ot \sf{snapshot}(E)$ \crm
\Cl \> \> {\bf for} $l \ot 0$ to $n - 1$ {\bf do} \crm
\Cl \> \> \> $(\mathbf{v}_0, \mathbf{v}_1, s) 
\ot E'[l + j - n].\textsf{read}_i()$ \crm
\Cl \> \> \> {\bf if} $s \neq \bot$ {\bf then} \` 
\Scomment{when $s$ is committed} \crm
\Cl \> \> \> \> $C[l] \ot \mathbf{v}_s$ \crm
\Cl \> \> \> {\bf else} \` \Scomment{when $s$ is uncommitted} \crm
\Cl \> \> \> \> $C[l] \ot \mathbf{v}_i$; $C[n] \ot \mathbf{v}_{1- i}$; 
$q \ot $FALSE \crm
\Cl \> \> \> {\bf endif} \crm
\Cl \> \> {\bf endfor} \crm
\Cl \> \> {\bf if} all slots in $E'[j-n, j-1]$ are committed {\bf then} \crm
\Cl \> \> \> $C[n] \ot \mathbf{m}(C)$ \crm
\Cl \> \> {\bf endif} \crm
\Cl \> \> return($q, C$) \crm
\Cl \> {\bf endprocedure} \crm
\crm
\Cl \> {\bf when} propose$_i(v)$ : \crm
\Cl \> \> {\bf for} $u \ot 0$ to $n-1$ {\bf do} \` 
\Scomment{Construction of initial configuration} \crm
\Cl \> \> \> $E[u].\textsf{submit}_i((v, 0), \mathrm{INIT}_u))$ \`
\Scomment{$\mathrm{INIT}_u$ is the initial state of $r_u$} \crm
\Cl \> \> {\bf endfor} \crm
\Cl \> \> $u \ot n$ \crm
\Cl \> \> {\bf loop} \crm
\Cl \> \> \> $(q, C) \ot \mathsf{getview}_i(u)$ \crm
\Cl \> \> \> {\bf if} $q =$ TRUE and $C$ is legitimate 
{\bf then} \crm
\Cl \> \> \> \> decide$_i(\mathsf{decode}(\mathbf{m}(C)))$ and {\bf exit} \crm 
\Cl \> \> \> {\bf endif} \crm
\Cl \> \> \> $E[u].\textsf{submit}_i(\mathsf{move}({\mathcal{A}_{(u \bmod n)}},
C, r_{(u \bmod n)},))$ ; $u \ot u + 1$ \crm
\Cl \> \> {\bf endloop} 
\end{tabbing}
}
\end{algorithm}

\subsection{Outline of the Correctness Proof}
\label{sec:Correctness}

We informally show how and why our simulation algorithm correctly
solves the consensus. We first introduce several notations:
Each slot consisting in the array $E$ is identified by its index.
We say that a process $p_i$ {\em enters} a slot $j$ (or finishes $j-1$) 
at $t$ if it takes $j$-th snapshot at $t$. The time when $p_i$ enters
slot $j$ is denoted by $t^i_j$.
Let $c(j)$ be the process ID submitting the committed value of slot $E[j]$
(say ``$p_{c(j)}$ commits $E[j]$'' or ``$p_{c(j)}$ is {\em committer} of 
$E[j]$'' in what follows), and $C^i_j$ be the observation result that $p_i$
obtains at the $j$-th main loop. We define 
$\alpha(j) = {j \bmod n}$ for short (i.e., $\alpha(j)$ is the ID of the 
robot to be activated at the $j$-th main loop). 
We also introduce {\em swap operator} $\pi_k$. For a given configuration 
$C$, we define $\pi_k{C}$ to be the configuration obtained by 
swapping two entries $C[k]$ and $C[n]$. By the definition, 
$\pi_n{C} = C$ clearly holds. 

Intuitively, the role of swap operators is to correct ``misunderstanding''
of uncertain slots. An example can be shown as follows: Let $t$ be the time
when the committer $p_0$ of slot $j$ takes $j$-th snapshot $E'$. Assume
a slot $g$ is uncommitted in $E'$ and $p_{1}$ commits both slots $g$ and 
$j+1$. Letting $(x_0, x_1)$ be two values submitted to $g$, since
$g$ is uncommitted in $E'$, the constructed observation result $C^0_j$
satisfies $C^0_j[\alpha(g)] = x_0$ and $C^0_j[n] = x_1$. On the other hand,
since $g$ is committed at the construction of $C^1_j$, 
$C^1_j[\alpha(g)] = x_1$ holds. In this case, we cannot have the execution connecting
$C^0_j$ and $C^1_{j+1}$, but the connection between
$\pi_{\alpha(g)}C^0_j$ and $C^1_{j+1}$ is possible. 

We describe $C \trans{x} C'$ if the activation of Byzantine robot $r_n$ 
followed by that of $r_{x}$ makes $C$ reach $C'$.
Then, it is assumed that the activation of Byzantine robot $r_n$ 
provides the best case movement. That is, $C \trans{x} C'$ 
implies that $C$ can reach 
$C'$ if the Byzantine robot appropriately moves after the 
activation of $r_x$.
The intended (but not attained) principle of our simulation 
is that one can organize the simulated execution 
$E = \pi_{\gamma_n}C^{c(n)}_n \trans{\alpha(n)} \pi_{\gamma_{n+1}}C^{c(n)}_{n+1} 
\trans{\alpha(n+1)} \cdots \trans{\alpha(j-1)} \pi_{\gamma_j}C^{c(j)}_j 
\trans{\alpha(j)} \pi_{\gamma_{j+1}}C^{c(j+1)}_{j+1}, \cdots$ for 
an appropriate sequence $\gamma_n, \gamma_{n+1}, \gamma_{n+2} \cdots$.
Unfortunately, that intention is broken in some critical case, 
which is explained below:  

\begin{quote}
Consider the situation where some process $p_i$ starts the $(j+n)$-th 
loop before the commitment of slot $j$. In this situation, $p_i$ has 
to make $r_{\alpha(j+n)}$ move in spite of the uncertainty of 
its location (because $E[j]$ is not committed yet and thus there 
are two possibilities $x_0$ and $x_1$ of $r_{\alpha(j)}$'s current 
local state). Then, $p_i$ chooses $x_i$ as if $E[j]$ already 
committed with value $x_i$. However, if $E[j]$ is actually 
committed with the value $x_{1 - i}$ by the other process 
$p_{1 - i}$, the inconsistency arises: We cannot construct 
a single execution that reflects two committed values of 
slots $j$ and $j+n$. 
\end{quote}

To achieve the consistency in the above scenario, we need to
correct the committers of $j$ and $j+n$ as if 
they are committed by the same value. That scenario and the correction
scheme is formalized by the notions of {\em critical slots} and 
{\em validators}, which are defined as follows:

\begin{definition}
Let $\mathcal{E} = [n+1, n+l]$ be a finite-length sequence of slots.
A slot $j$ is {\em critical} in $\mathcal{E}$ 
if it is uncommitted at $t^{c(j)}_{j+n}$ 
in the $(j+n)$-th snapshot taken by $p_{c(j)}$ and $c(j) \neq c(j+n)$ holds. 
\end{definition}

\begin{definition}
Let $\mathcal{E} = [n+1, n+l]$ be a finite-length sequence of slots.
The {\em validator} $p_{\mathit{val}(j)}$ of slot $j$ in $\mathcal{E}$
is defined as, \begin{equation*}
p_{\mathit{val}(j)} = \left\{
\begin{array}{ll}
p_{\mathit{val}(j+n)}, & \mbox{if $j$ is critical and $j \leq l$}\\
p_{c(j)}, & \mbox{otherwise}
\end{array}
\right. 
\end{equation*}
\end{definition}

Note that the criticality and validator of 
each slot is determined for a fixed finite-length
$\mathcal{E}$. Thus, if we consider a longer sequence $\mathcal{E}'$ 
obtained by adding a postfix sequence into $\mathcal{E}$, those can change because the criticality and the validators of last $n$ slots 
$[j - n + 1, j]$ depends on the following $n$ slots $[j + 1, j+n]$.
Intuitively validators are the processes whose observation results constitute 
the simulated execution.  Actually, the {\em simulated configuration}
corresponding to each slot is defined as follows:

\begin{definition}
Given a finite-length sequence $\mathcal{E}$ of slots, the {\em simulated
configuration} $C_j$ of slot $j$ for $\mathcal{E}$ is defined as follows:
\begin{itemize}
\item $C_j = C^{\mathit{val}(j)}_j$ if no slot in $[j-n, j-1]$ is 
uncommitted at $t^{\mathit{val}(j)}_j$.
\item If a slot $k \in [j - n, j-1]$ is uncommitted at 
$t^{\mathit{val}(j)}_j$, $C_j = \pi_\gamma C^{\mathit{val}(j)}_j$ 
for appropriate $\gamma$ ($\in \{\alpha(k), n\}$) such that 
$\pi_\gamma C^{\mathit{val}(j)}_j[\alpha(k)] = v^{\mathit{val}(j)}_k$
holds. 
\end{itemize}
\end{definition}

The key lemma of our reduction scheme is the sequence of 
simulated configurations (say {\em simulated execution}) 
constitutes a possible execution.

\begin{lemma} \label{lmaAdmissibleExecution}
Given a finite-length sequence of slots $\mathcal{E}$ and any 
$j$ that is smaller than the length of $\mathcal{E}$, 
$C_j \trans{\alpha(j)} C_{j+1}$ holds. 
\end{lemma} 

The above lemma implies that there exists an
execution $C_n, C_{n+1}, C_{n+2}, \cdots $ under the
$(n-1)$-bounded centralized scheduler. Informally, the uniform 
agreement and termination properties are deduced from this lemma
and the correctness of the gathering algorithm. However, the validity 
must be considered more carefully. An important notice is that 
the validity property is strongly related to the way of defining 
function {\sf decode}. In the rest of this subsection, we state 
the appropriate definition of function {\sf decode} guaranteeing 
the validity of {\sf ConsensusToGathering}. 

Let us consider the situation where both $p_0$ and $p_1$ propose 
the same value (assume zero). Then all robots are placed on $(0, 0)$ 
initially in the simulation. As we mentioned in the previous section, 
however, it does not implies that all robots are gathered at $(0, 0)$ 
because they are non-oblivious. Let $C$ be the configuration where 
all robots are placed on $(0, 0)$. To guarantee the validity, we 
need to decode the point of gathering in any simulated execution 
starting from $C$ to zero. Fortunately since $p_0$ and $p_1$ has 
a common proposal, it is ensured that the simulated execution of 
$\mathcal{A}$ is uniquely determined: They submit the same value 
to each slot in $[0..n-1]$ From the common value commitment property 
of {\sf Slot} objects, any slot in $[0..n-1]$ is immediately committed 
when at least one process finishes it. Thus, when a process enters slot $n$, 
all slots of $[0..n-1]$ has been committed. That is, 
the configurations constructed at slot $n$ are the same among $p_0$
and $p_1$, and thus the submissions of $p_0$ and $p_1$ at slot $n$ are 
also the same. Inductively we can conclude that the submission values 
of two processes to any slot are the same. That is, the simulated execution 
(and thus the point of gathering) is uniquely determined from the initial 
configuration. Let $\mathbf{x}$ be the point of gathering corresponding 
to $C$. Provided {\sf decode}$(\mathbf{v})$ as the function 
returning zero if $\mathbf{v} = \mathbf{x}$ and one otherwise, we 
can guarantee that the algorithm decides zero if all processes 
propose zero. For the case that all processes propose one, 
we need to show that the gathering is never achieved at $\mathbf{x}$. 
It can be proved as follows: Let $C'$ be the configuration where all 
robots are placed on $(1, 0)$. By the same reason as the case of 
all proposing zero, the simulated execution starting from $C'$ is 
also uniquely determined. Furthermore, its activation schedule is 
completely same as the case of all proposing zero. Since each robot 
does not aware of global coordinates, it follows that the robot is 
gathered at $\mathbf{x} + (1, 0)$ in the simulated execution starting 
from $C'$.

The above argument clearly provides the validity of the consensus, and thus
the following theorem is obtained.

\begin{theorem} \label{thmImpossibilityGathering}
The algorithm {\sf ConsensusToGathering} correctly solves the 
consensus problem, and thus
there exists no gathering algorithm tolerating up to one Byzantine 
fault even if we assume agreed direction of local coordinate systems, 
instantaneous movements, $n$-bounded centralized scheduler, and 
non-oblivious and non-uniform robots.
\end{theorem}

\section{General Formation}
\label{secGeneralPattern}

While the algorithm {\sf ConsensusToGathering} is constructed for 
leading the impossibility of Byzantine gathering,  it is easy to use 
its scheme for obtaining the impossibility of a more general class 
of formation problems.

We first formally define the general formation problem. A Byzantine 
formation problem on $n + f$ robots is defined by a family $\mathcal{F}$ 
of multisets of locations with cardinality $(n+f)$. An algorithm 
$\mathcal{B}$ solves the Byzantine formation problem $\mathcal{F}$ if
in any execution of $\mathcal{B}$ all correct robots eventually form 
and keep the location set that is a subset of an element in $\mathcal{F}$.
Clearly, it is not possible to show the impossibility of any formation
$\mathcal{F}$. For example, if $\mathcal{F} = (\mathcal{R} \times {R})^{n+f}$
(i.e., any multiset of $n+f$ locations belongs to $\mathcal{F}$), 
$\mathcal{F}$ can be solved trivially. In the following argument,
we consider the family of patterns to which we can apply our reduction
technique, and explain how we apply it. 

We only consider the case of $f=1$. We define a $1$-neighborhood
relation between two location configurations $P$ and $P'$, which 
holds if and only if $|P \cap P'| \geq n$. The transitive 
closure of $1$-neighborhood relation is denoted by $\sim$. Given 
a formation problem $\mathcal{F}$, we define its {\em $1$-neighborhood 
extension} $\mathcal{F}^1 = \{P' | \exists P \in \mathcal{F}, P' \sim P\}$.
Since the relation $\sim$ is an equivalence relation, we can define 
an equivalent class over $\mathcal{F}^1$, which is denoted by 
$[\mathcal{F}^1]$. Given $P = \{\mathbf{v}_0, \mathbf{v}_1, \cdots, 
\mathbf{v}_{n+1}\}$, we define $P + \mathbf{x} = 
\{\mathbf{v}_0 + \mathbf{x}, \mathbf{v}_1 + \mathbf{x}, \cdots, 
\mathbf{v}_{n+1} + \mathbf{x} \}$. We first define the 
formation problems we can handle in our reduction algorithm.

\begin{definition}
A formation problem $\mathcal{F}$ is said to be {\em bivalent} if
there exists $\mathbf{x}_P$ for any $P \in \mathcal{F}$ such that
$P$ and $P + \mathbf{x}_P$ belong to different classes in $[\mathcal{F}^1]$.
\end{definition}

It can be shown that many well-known pattern formation problems 
(circle, line, and so on) are bivalent. We present several examples
in the appendix.

We introduce the modification of the algorithm {\sf ConsensusToGathering} to 
lead the impossibility of bivalent pattern formation problems. The framework 
is completely same as the reduction to gathering. Each simulator process 
first tries to place all robots on some coordinates according to its proposal, 
and run the simulation. Finally, the process decodes a decision value when 
the simulation reaches a legitimate configuration. The points to be 
addressed are \emph{(i)} the definition of functions 
$M(C)$ and $\mathbf{m}(C)$,  and \emph{(ii)} the locations of robots 
initially placed and the definition of 
function {\sf decode}. We explain the modification of the simulation 
algorithm for those points:

\paragraph{Defining $M(C)$ and $\mathbf{m}(C)$}
The function $M(C)$ is defined as $M(C) = 
\max_{P \in \mathcal{F}}|L(C) \cap P|$. Let $P'$ be the pattern 
in $\mathcal{F}$ maximizing $|L \cap P|$ (if two or more patterns 
in $\mathcal{F}$ maximize it, an arbitrary one is deterministically chosen). 
We define $\mathbf{m}(C)$ as a coordinate in $P' \setminus L$, which 
is also chosen deterministically if $|P' \setminus L| >1$. 

\paragraph{Initial location and decoding function {\sf decode}} 
For the proposal zero, we define the initial location configuration of $n$ 
robots $((0, 0), (1, 0), (2, 0), \cdots, (n-1, 0))$. That is, the 
process proposing zero submits value $((i, 0), \mathrm{INIT}_i)$ to 
slot $i$ ($\mathrm{INIT}_i$ is the initial local state of robot $r_i$). 
By the same argument as the case of gathering, for that initial 
configuration, we can uniquely determine the legitimate configuration $C$ 
that the simulated execution eventually reaches. From the definition 
of bivalent formation problems, there exists a vector $\mathbf{x}$ such 
that $L(C) + \mathbf{x}$ belongs to a class of $[\mathcal{F}^1]$ different 
from what $L(C)$ belongs to. We define the initial location configuration 
for proposal one as $((0, 0) + \mathbf{x}, (1, 0) + \mathbf{x}, 
(2, 0) + \mathbf{x}, \cdots, (n-1, 0))$. We further define {\em decode} 
as the function from a (legitimate) configuration to $\{0, 1\}$.
It returns zero if the given configuration belongs to a class of 
$\mathcal{F}^1$ where $L$ belongs, and one otherwise.

We can prove similarly as Theorem \ref{thmImpossibilityGathering} that 
the modified algorithm correctly solves the consensus. Consequently
the following theorem is obtained.

\begin{theorem} \label{thmImpossibilityPattern}
Any bivalent formation problem is unsolvable in the system with $f=1$ 
even if we assume a common orientation of local coordinate 
systems, instantaneous movements, $n$-bounded centralized scheduler, 
and non-oblivious and non-uniform robots.
\end{theorem}

\bibliographystyle{abbrv}

\bibliography{gathering-consensus}

\newpage

\appendix

\section{Correctness Proof of Lemma \ref{lmaAdmissibleExecution}}

In the following argument, we use the notation 
$t_j = t^{\mathit{val}(j)}_j$ and $v_j = v^{\mathit{val}(j)}_j$ for short. 
We first show two lemmas that are used in the main proof:

\begin{lemma} \label{lmaValidatorSubmission}
For any $j$, $C_{j+1}[\alpha(j)] = v_j$ holds.
\end{lemma}

\begin{proof}
This trivially holds from the definition of validators 
and simulated configurations. \qed
\end{proof}

\begin{lemma} \label{lmaValidationOrder}
$t_j < t_{j+1}$.
\end{lemma}

\begin{proof}
Suppose for contradiction that $t_j > t_{j+1}$ holds. 
Then $p_{\mathit{val}(j)} \neq p_{\mathit{val}(j+1)}$ clearly
holds and thus $p_{\mathit{val}(j+1)}$ finishes slot 
$j$ before $p_{(1 - \mathit{val}(j+1))}$ enters $j$. 
It implies that $p_{\mathit{val}(j)}$ is not the committer of $j$
and thus $j$ is critical. However, to make $j$ critical, 
$p_{\mathit{val}(j)}$ must enter $j+n$ at $t_{j+1}$ or earlier
because $j$ has been committed at $t_{j+1}$ or earlier, which 
contradicts $t_j > t_{j+1}$. \qed
\end{proof}

By using these lemmas, we prove Lemma \ref{lmaAdmissibleExecution}.

\begin{proof}
From Lemma \ref{lmaValidatorSubmission}, it suffices to show
that $C_j[x] = C_{j+1}[x]$ holds for any $x \in [0, n] 
\setminus \{\alpha(j), n\}$. Since $C_j[\alpha(k)] = 
C_{j+1}[\alpha(k)]$ holds if slot $k$ has the same
status at $t_j$ and $t_{j+1}$, only the scenario we have to 
consider is that the status of some slot $k \in [j - n + 1, j - 1]$
changes from uncommitted to committed between $t_j$ and $t_{j+1}$ (recall 
$t_j < t_{j+1}$ from Lemma \ref{lmaValidationOrder}).
We show that $C_j[\alpha(k)] = C_{j+1}[\alpha(k)]$ holds in this scenario,
which is sufficient to prove the lemma. Suppose for contradiction
that $C_j[\alpha(k)] \neq C_{j+1}[\alpha(k)]$. From the definition
of simulated configurations, we have $C_j[\alpha(k)] = v^{\mathit{val}(k)}_k$
and $C_{j+1}[\alpha(k)] = v^{c(k)}_k$. Thus, $\mathit{val}(k) \neq c(k)$, which
implies that $k$ is critical. Then $p_{\mathit{val}(k)}$ must stay slot $k$
until $p_{1 - \mathit{val}(k)}$ enter $k+n (> j+1)$. It however contradicts
the fact $k$ is committed at $t_{j+1}$ because $t_{j+1}$ is the 
time when $p_{1 - \mathit{val}(k)}$ commits $j+1$. \qed
\end{proof}

\section{Proof of Theorem \ref{thmImpossibilityGathering}}

The theorem is proved by showing that the algorithm 
{\sf ConsensusToGathering} guarantees the uniform agreement, termination, 
We say that a simulated configuration $C$ is {\em semi-legitimate}
if there exists $\gamma \in [0, n-1]$ such that $\pi_{\gamma}C$ 
is legitimate. Note that any legitimate configuration is semi-legitimate
because of $\pi_{\gamma_n}C = C$. The following lemma takes an important 
role in the following proof.

\begin{lemma} \label{lmaN-1Gathered}
Let $\mathcal{A}$ be an an arbitrary Byzantine gathering algorithm
tolerating up to one fault, and $\mathcal{E} = C'_0, C'_1, C'_2, \cdots 
C'_j$ be an arbitrary 
execution of $\mathcal{A}$ such that $C'_0$ is semi-legitimate and 
$\mathbf{r}_n(C'_k) = \mathbf{m}(C'_k)$ or 
$\mathbf{r}_n(C'_{k - 1})$ holds for any $k$ 
($1 \leq k \leq j$).
 Then, for any configuration $C'_h$ in $\mathcal{E}$,  
$M(C'_h) \geq n$ and $\mathbf{m}(C'_0) = \mathbf{m}(C'_h)$ hold.
\end{lemma}

\begin{proof}
We show that $M(C'_1) \geq n$ and $\mathbf{m}(C'_0) = 
\mathbf{m}(C'_1)$ holds. About the following configurations
$C'_2, C'_3 \cdots $, we can inductively apply the same argument as
$C'_1$. If all correct robots are located on $\mathbf{m}(C'_0)$ 
at $C'_0$, $C'_0$ is already legitimate and the lemma clearly holds.
Otherwise, we have $\mathbf{r}_n(C'_0) = \mathbf{m}(C'_0)$ and 
$M(C'_0) = n$. We show that any robot on $\mathbf{m}(C'_0)$ never moves 
during the transition from $C'_0$ to $C'_1$. Let 
$r_\gamma$ be the robot not on $\mathbf{m}(C'_0)$ at $C'_0$. Since
$C'_0$ is semi-legitimate, the configuration $\pi_\gamma C'_0$ is 
a legitimate configuration. Then the robots on $\mathbf{m}(C'_0)$ 
cannot distinguish $C'_0$ and $\pi_\gamma C'_0$, and thus 
they never change their positions. Since we assume that 
$\mathbf{r}_n(C'_1) = \mathbf{m}(C'_1) (= \mathbf{r}_n(C'_0))$, 
we can conclude all robots on $\mathbf{m}(C'_0)$  keep their 
positions at $C'_1$. The lemma is proved. \qed
\end{proof}

\begin{lemma}[Uniform Agreement] \label{lmaAgreement}
If $p_0$ and $p_1$ decide, their decision values are same. 
\end{lemma}

\begin{proof}
Without loss of generality, we assume that $p_0$ decides at a 
slot earlier than or equal to $p_1$'s . Let $j$ and $j+h$ be the slots 
where $p_0$ and $p_1$ decide respectively. 
Since $p_0$ exits at the beginning of slot $j$, all slots
$[j, j+h]$ are solely committed by $p_1$. From the definition,
those slots are not critical. It follows that 
$\mathbf{m}(C^1_k) = \mathbf{m}(C_k)$ holds for any $k \in [j, j+h]$.
In addition, we can obtain $\mathbf{m}(C^1_j) = \mathbf{m}(C_j)$ 
because all slots of $[j - n, j - 1]$ are committed at the 
construction of $C^1_j$ (from the definition of simulate configuration
$C_j$, each entry $C_j[x]$ for any $x \in [0, n]$ is the committed
value of the corresponding slot in $[j -n, j -1]$). 
Assume that $C_j$ is semi-legitimate and $\mathbf{r}_n(C_k) = 
\mathbf{m}(C_k)$ or $\mathbf{r}_n(C_{k - 1})$ 
holds for any $k \in [j+1, j+h]$. Then, from
Lemma \ref{lmaN-1Gathered}, we can conclude $\mathbf{m}(C^0_j) 
= \mathbf{m}(C_j) = \mathbf{m}(C_{j+h}) = \mathbf{m}(C^1_{j+h})$ 
and thus the lemma holds. 

The rest of the proof is to show that those assumptions hold. 
We first show that the first one holds: Since
$C^0_j$ is legitimate and all slots of $[j - n, j -1]$ are committed at
$t^0_j$, $C^1_j$ is legitimate if those slots are already committed at
$t^1_j$. Otherwise, letting $l$ be the slot in $[j - n, j - 1]$ 
uncommitted at $t^1_j$, $\pi_{l}C^1_j$ or $C^1_j$ is legitimate, which
implies that $C^1_j$ is semi-legitimate.

Next we look at the second assumption. We give the proof for 
the case of $k = j+1$. 
For any following slot $k > j+1$, we can inductively prove the assumption 
in the same way as that case.
Consider the following two cases: 1) If all slots in $[k - n, k - 1]$ 
are committed at $t^1_k$, we obviously have $\mathbf{r}_n(C^1_{k}) 
= \mathbf{m}(C^1_{k})$. 2)A slot $l \in [k - n, k - 1]$ is uncommitted
at $t^1_k$, $l$ is uncommitted during $[t^1_j, t^1_k]$ because any 
slot processed after $t^1_j$ is immediately committed when $p_1$ finishes it. 
This implies that the status of $l$ is the same at $t^1_k$
and $t^1_{k - 1}$ and thus $\mathbf{r}_n(C^1_{k}) = \mathbf{r}_n(C^1_{k - 1})$
holds. The lemma is proved. \qed
\end{proof}

\begin{lemma}[Termination] \label{lmaTermination}
Each process $p_i$ eventually decides unless it crashes. 
\end{lemma}

\begin{proof}
We first show that at least one of $p_0$ and $p_1$ eventually 
decides. Lemma \ref{lmaAdmissibleExecution} implies that we can 
have an admissible execution $\mathcal{E} = C_0, C_1, \cdots C_j$ for 
sufficiently large $j$ where the last $n$ slots are eventually 
committed (that is, no process crashes during the processing of 
those slots) and the simulated configurations $C_{j - n + 1}, 
C_{j - n + 2}, 
\cdots, C_{j}$ corresponding to them are converged to 
a legitimate one. Let $\mathbf{m}$ be the point of gathering,
that is, $\mathbf{m} = \mathbf{m}(C_{j - n + 1}) =  
\mathbf{m}(C_{j - n + 2}) =, \cdots = \mathbf{m}(C_j)$.
From the definition of simulated executions the last $n$ configurations 
are the observation results by the committer of each slot.
Since $C_k$ for any $k \in [j - n+1, j]$ is legitimate, its committed value 
is $\mathbf{m}$. Let $p_i$ be the process such that no process enters 
$j+1$ after $p_i$ does. Since all slots in $[j - n + 1, j]$ have 
been committed with $\mathbf{m}$ when $p_i$ enters $j+1$, 
$p_i$ decides at $j+1$. 

The remaining part of the proof is to show
that the decision of $p_i$ implies that of $p_{1 - i}$. Since 
$p_i$ exits the algorithm at the beginning of $j+1$, its behavior from
the viewpoint of $p_{1 - i}$ is equivalent to the crash at $j+1$. 
If $p_i$ crashes, $p_{1 - i}$ necessarily decides because we have 
already proved that at least one process eventually decides. Consequently,
we can conclude that $p_{1 - i}$ eventually decides after the decision of 
$p_i$. \qed
\end{proof}

\section{Examples of Bivalent Formation Problems}

We show that two well-known formation problems, circle and line,
belongs to the class of bivalent formation problems.

\begin{example}
The {\em circle formation} is the problem requiring that all robots
are placed on different locations on the boundary of a common
circle\footnote{Exactly, we consider the {\em non-uniform} circle 
formation problem. A stronger variant is {\em uniform} circle formation, 
which must guarantees that all robots
are placed evenly on the boundary of a common circle.}. The 
specification $\mathcal{F}_{\mathrm{circle}}$
of this problem can be stated as follows: $\mathcal{F}_{\mathrm{circle}}
= \{ \{ \mathbf{v}_0, \mathbf{v}_1, \cdots \mathbf{v}_{n+f} \}| 
\forall \mathbf{v_i}, \mathbf{v_j} : \mathbf{v}_i \neq \mathbf{v}_j \wedge 
\exists \mathbf{c}, r \in \Real : \forall \mathbf{v}_i : 
|\mathbf{v}_i - \mathbf{c}| = r \}$.
\end{example}

\begin{example}
The {\em line formation} is the problem requiring that all robots
are placed on different locations on a common line, which can be 
specified as $\mathcal{F}_{\mathrm{line}} = \{ \{ \mathbf{v}_0, \mathbf{v}_1, 
\cdots \mathbf{v}_{n+f} \} | 
\forall \mathbf{v_i}, \mathbf{v_j} : \mathbf{v}_i \neq \mathbf{v}_j \wedge
\exists a_2, a_3, \cdots, a_{n+f} \in \Real : 
\forall i \in [2, n+f] : \mathbf{v}_i - \mathbf{v}_0 
= a_i (\mathbf{v}_1 - \mathbf{v}_0) \}$.
\end{example}

\begin{theorem}
If $n + 1 > 4$, the circle formation and the line formation 
are bivalent.
\end{theorem}

\begin{proof}
We only show the proof for the circle formation because the 
bivalency of the line formation can be proved in the same way as 
the circle. From the definition of the circle formation problem, we can 
associate the circle containing at least $n$ robots to each pattern 
$P \in \mathcal{F}_{\mathrm{circle}}^1$, which are denoted by 
$\mathit{cir}(P)$. Let $P$ be an arbitrary pattern 
in $\mathcal{F}_{\mathrm{circle}}$. 
To prove the theorem, it suffices to show that there exists 
a vector $\mathbf{x}$ such that $P' = P + \mathbf{x}$ satisfies 
$P' \not\sim P$. Suppose for contradiction
that $P \sim P'$ holds for any $\mathbf{x}$. Since $P \sim P'$ 
holds, we have a chain $P = P_0 \sim P_1 \sim P_2 \cdots \sim P_k = P'$ 
where $P_i \in \mathcal{F}_{\mathrm{circle}}^1$ and 
$|P_i \cap P_{i+1}| \geq n$ hold for any $i$ ($0 \leq i \leq k$). 
Since $\mathit{cir}(P_0) \neq \mathit{cir}(P_k)$ clearly holds, 
there necessarily exists $h$ satisfying $\mathit{cir}(P_h) 
\neq \mathit{cir}(P_{h+1})$. However, it contradicts the fact that 
$|P_h \cap P_{h+1}| \geq n > 3$ because at most three robots can 
be placed on the intersection of $\mathit{cir}(P_h)$ and 
$\mathit{cir}(P_{h+1})$ in $P_h$ and $P_{h+1}$ (recall that the 
circle formation requires that all correct robots must be 
located on different positions in legitimate configurations).
\qed
\end{proof}

We further present a formation problem that is not bivalent.

\begin{example}
The {\em $2$-gathering} is the problem requiring that all robots are
placed on at most two locations. It is specified as 
$\mathcal{F}_{\mathrm{2gat}} = \{ \mathbf{v}_0, \mathbf{v}_1, 
\cdots \mathbf{v}_{n+f} | \exists \mathbf{x}_0, \mathbf{x}_1 :
\forall i \in [0, n+f] : \mathbf{v}_i = \mathbf{x}_0 \vee \mathbf{v}_i
= \mathbf{x}_1 \}$.
\end{example}

\begin{theorem}
The $2$-gathering problem is not bivalent.
\end{theorem}

\begin{proof}
We show the theorem by showing that $[\mathcal{F}_{\mathrm{2gat}}]$
consists of a single class. That is, for any two patterns 
$P, P' \in \mathcal{F}_{\mathrm{2gat}}$, we prove $P \sim P'$.
Let $\{ \mathbf{p}_0, \mathbf{p}_1\}$, and 
$\{\mathbf{p}'_0$, $\mathbf{p}'_1\}$ be the set of locations constituting
$P$ and $P'$ respectively. Taking two patterns $Q$ and $Q'$ in 
$\mathcal{F}_{\mathrm{2gat}}$ where all robots are placed in $\mathbf{p}_0$
and $\mathbf{p}'_0$, we can easily show that $P \sim Q \sim Q' \sim P'$
holds: One-by-one replacement of all robots at $\mathbf{p}_1$ to $\mathbf{p}_0$
transforms $P$ into $Q$, which implies $P \sim Q$. The relations
$Q \sim Q'$ and $Q' \sim P$ can be obtained similarly. Consequently,
we have $P \sim P'$. \qed
\end{proof}

Because of the above theorem, we cannot prove the impossibility of 
the $2$-gathering problem from our reduction. We conjecture that the 
$2$-gathering is solvable if we assume that $f=1$ and the agreement of
the orientations of local coordinate systems.

\end{document}